\documentclass[conference,a4paper]{IEEEtran}

\usepackage{amsmath, amsthm, amsfonts}
\IEEEoverridecommandlockouts

\usepackage{graphicx}

\newtheorem{theorem}{Theorem}
\newtheorem{lemma}{Lemma}
\newtheorem{corollary}{Corollary}

\def\be{\begin{equation}}
\def\ee{\end{equation}}
\def\ie{\emph{i.e.}, }
\def\iid{\emph{i.i.d.}}
\def\E{{\mathbb E}}
\def\a{\alpha}
\def\C{\bar{C}}

\def\D{\mathcal{D}}

\def\S{\mathcal{S}}
\def\P{{\mathbb P}}
\def\Sc{\mathcal{S}^\mathsf{c}}
\def\Dc{\mathcal{D}^\mathsf{c}}

\begin{document}
\title{Information Theoretic Cut-set Bounds on the Capacity of Poisson Wireless Networks}

\author{\IEEEauthorblockN{ Georgios Rodolakis\IEEEauthorrefmark{1}}
\IEEEauthorblockA{\IEEEauthorrefmark{1}Information Technologies Institute,  Centre for Research and Technology, Greece, rodolakis@iti.gr}
\thanks{This work was supported in part by the EU Marie Curie career integration grant INFLOW (FP7-294236).}
}

\maketitle

\begin{abstract}
This paper presents a stochastic geometry model for the investigation of fundamental information theoretic limitations in wireless networks.
We derive a new unified multi-parameter cut-set bound on the capacity of networks of arbitrary Poisson node density, size, power and bandwidth, under fast fading in a rich scattering environment.
In other words, we upper-bound the optimal performance in terms of total communication rate, under any scheme, that can be achieved between a subset of network nodes (defined by the cut) with all the remaining nodes. 
Additionally, 
we identify four different operating regimes, depending on the magnitude of the long-range and short-range signal to noise ratios. Thus, we confirm previously known scaling laws ({\it e.g.}, in bandwidth and/or power limited wireless networks), and we extend them with specific bounds. Finally, we use our results to provide specific numerical examples.
\end{abstract}

\section{Introduction}
The investigation of fundamental capacity limits of multi-node wireless networks is an open problem in information theory 
which consistently attracts the attention of researchers in recent years, 
as it is a difficult question with great potential practical interest. A way to approach the problem is to study the more restricted situation where wireless nodes are placed according to a spatial distribution (usually 2-dimensional), with a simplified propagation model. Furthermore, we can study the scaling behavior of the total network capacity at the limit where the number of nodes tends to infinity.

In this context, initial investigations 
focused on scaling laws with specific communication strategies, such as multi-hopping~\cite{GK},
providing important insights on the fundamental limits of wireless networks. Several works studied information-theoretic scaling laws, independent from the communication strategy. However, results usually provide only an asymptotic order for the network capacity ({\emph e.g.},~ \cite{snr_ot,dense}).
Most importantly, these results provided insights on cooperation schemes with almost optimal scaling behavior. For instance, 
\cite{dense} shows that {\it dense} (\ie of fixed area and increasing density) and {\it extended} (\ie of fixed density and increasing area) networks exhibit qualitatively different scaling behaviors with regard to the total network capacity (linear and sub-linear or square root increase with respect to the number of nodes, respectively).
In contrast, real networks have a fixed area and density and, in this sense, such scaling laws are of limited use for practical purposes. This limitation has been partially addressed with an insightful extension in~\cite{snr_ot}, where it is shown that important parameters defining the asymptotically optimal operating regime of a wireless network are the short range and long range signal to noise ratios (SNR).

This paper focuses on the derivation of fundamental cut-set bounds on the capacity of wireless networks. 
Taking a cut partitioning the network into two parts, we bound above the sum of the rates of communication passing through the cut, under any communication strategy.
We rely on a Poisson network model, which we analyze using a stochastic geometry methodology, under a rich and symmetric fading environment (which we describe in detail in Section~\ref{sec:model}).

To motivate our approach, consider a point-to-point channel of bandwidth $W$ Hertz, and additive white Gaussian noise (AWGN) with power spectral density $\frac{N}{2}$. The channel capacity is given by the simple formula: $C= W \cdot \log \left( 1+ \frac{P}{N W} \right)$, with $P$ the received power.
This formula identifies two operating regimes: for low SNR, we have $C\sim \frac{P}{N}$, and the capacity is power-limited;
for high SNR, we have $C\sim W \log  \frac{P}{N W}$, and the capacity is essentially bandwidth-limited. 

In this paper, we provide such a unified formula, as an upper bound on cut-set capacities in Poisson wireless networks. 
We then show that its asymptotic behavior is richer, but not too complicated (providing four different asymptotic regimes).
Identifying such operating regimes is of great usefulness for the design of efficient communication strategies.
We evaluate the cut-set capacity of approximately circular cuts of arbitrary radius, for arbitrary values of all the other network parameters, such as the node density, transmit power, channel bandwidth, noise spectral density. 
We confirm that our results are in agreement with the previously known scaling laws identified in~\cite{snr_ot}. 
Additionally,  
we derive specific numeric bounds that capture the continuous transitions between different operating regimes, complementing previous related work.

In Section~\ref{sec:model}, we introduce our channel and network model and we discuss our main results.
In Section~\ref{sec:cutsetbound}, we
prove our general cut-set bound (Theorem~\ref{th:approx_cutset}),
and we identify 
asymptotic 
operating regimes (Corollary~\ref{cor:asympt}).
We provide specific numerical examples in Section~\ref{sec:numerical}.

\section{Model and Main Results}
\label{sec:model}

\subsection{Channel Model}

We consider a network where nodes are equipped with wireless transceiver capabilities (with a single transmit and a single receive antenna) and transmissions occur at discrete times $t=1,2,\ldots$. 
Communication takes place over a flat channel of bandwidth $W$ Hertz around a carrier frequency $f_c$, with $f_c \gg W$.
Node positions remain fixed during a channel use.
Let $r_{ij}$ denote the distance between nodes $i$ and~$j$. 
The received power decays with the distance in a power law, with path loss exponent $\a > 2$.

We assume a fast-fading model.
We denote $H_{ij}[t]$ the complex base-band equivalent channel gain for transmissions from node $j$ to node $i$, at time $t$.
The gains depend on the distance between the node positions, and on the channel fading.
The channel gain $H_{ij}[t]$ has the form:
\be
\label{eq:Hij}
H_{ij}[t] = r_{ij}^{-\frac{\a}{2}} \cdot h_{ij}[t],
\ee
where  $r_{ij}^{-\frac{\a}{2}}$ models path loss, and $h_{ij}[t]$ is a stationary and ergodic random process that models channel fluctuations due to frequency flat fading. Without loss of generality, we let $\E[|h_{ij}[t]|^2] =1$.

We also make the two following modeling assumptions. First, the $h_{ij}[t]$'s are symmetric, \ie $h_{ij}[t]$ has the same distribution as  $-h_{ij}[t]$ (this implies a zero mean). Second, the $h_{ij}[t]$'s are independent for different $i,j$.

Our model is intended as an approximation of a rich scattering environment with a far-field assumption, \ie the node separation distance is sufficient for the channel independence and symmetry assumptions (and path loss model) to be realistic. It includes as special cases Rayleigh fading, as well as the \iid random phase model used in~\cite{snr_ot,dense}. In contrast, we do not model  non-zero-mean or correlated channel gains.

We denote $X_i[t]$ the symbol transmitted by node $i$ at time~$t$.
All nodes have an equal average power budget of $P$ Watts, \ie for all $i$, $\E[|X_i[t]|^2] \leq \frac{P}{W}$ Joules per symbol.
The signal received by node $i$ at time $t$ is given by
\be
Y_i[t] = \sum_{j \neq i} H_{ij}[t] \cdot X_j[t] + Z_i[t],
\label{eq:Yi}
\ee
where $Z_i(t)$ is white circularly symmetric complex Gaussian noise of power spectral density $N$ Watts per Hertz (\ie the real and imaginary parts each have variance $\frac{N}{2}$ per symbol). 

\subsection{Cut-set Capacity Bound}
\label{sec:modelbound}
We consider a cut partitioning the network into two complementary sets of nodes, denoted $\S$ and $\Sc$. 
We are interested in bounding above the sum of the rates of communication passing through the cut from $\Sc$ to $\S$, with arbitrary one-to-one source-destination pairings. 
The total rate is bounded above by the cut-set capacity $C_{\S,\Sc}$, defined as the maximum of the mutual information between transmitted and received symbols,
over all possible distributions of the transmitted symbols that satisfy the maximum power constraint.

Equivalently, the cut-set capacity  $C_{\S,\Sc}$ corresponds to the capacity of the multiple-input multiple-output (MIMO) channel between nodes in $\Sc$ and nodes in $\S$, with a per-antenna power constraint $P$.
Under the fast-fading model, the ergodic MIMO capacity in nats per second equals:
\be
\label{eq:ergodicMIMO}
C_{\S,\Sc} = \max_{\substack{Q \geq 0\\ \E[Q_{jj}] \leq P ,\forall j \in \Sc}} \E_H\left[W \log \det \left(I+\frac{1}{N W} H Q H^\ast \right)\right],
\ee
where $H_{ij}=r_{ij}^{-\frac{\a}{2}} h_{ij}$, and $Q$ is the positive semi-definite covariance matrix of the transmitted signal vector.

The factor $r_{ij}^{-\frac{\a}{2}}$ is assumed to be known at both the receiver and the transmitter, as the node positions are fixed.
The realization of $h_{ij}$ is just known at the receiver, whereas the transmitter only knows the channel distribution.

In our channel model, since the $h_{ij}$'s are independent and symmetrically distributed, the MIMO capacity formula can be simplified, based on~\cite[Corollary 1c]{abbe}; the input covariance matrix that maximizes the capacity is diagonal with all entries equal to the power constraint $P$:
\be
\label{eq:symmetricMIMO}
C_{\S,\Sc} = \E_{H}\left[W \log \det \left(I+\frac{P}{N W} H H^\ast \right)\right].
\ee
Therefore, in the optimal communication strategy, the transmit nodes send independent signals at full power, and there is no need to do any sort of transmit beamforming.

\paragraph*{Remark} We note that this upper bound is also valid in a general channel model (dropping the symmetry and independence assumptions) with arbitrary fading, under the condition that nodes may only transmit independent signals.

Using Hadamard's inequality ($I+\frac{P}{N W} H H^\ast$ is positive semi-definite), we obtain the upper bound:
\begin{align}
\label{eq:CMISO}
C_{\S ,\Sc} 
\leq \E_H \left[ W \sum_{i \in \S} \log  \left(1+\frac{P}{N W} \sum_{j \in \Sc} |H_{ij}|^2 \right) \right],
\end{align}
\ie the upper bound is the sum of the capacities of the multiple-input single-output (MISO) channels between nodes in $\Sc$ and each node in $\S$, with independent transmissions. 

Finally, using Jensen's inequality (the function $\log(1+x)$ is concave), and since $\E( |H_{ij}|^2)= r_{ij}^{-\a}$, we have:
\be
\label{eq:MISOJensen}
C_{\S ,\Sc} \leq W \sum_{i \in \S} \log \left(1+\frac{P}{N W} \sum_{j \in \Sc} r_{ij}^{-\a} \right),
\ee
which only depends on the geometry of the network (and not on the fading distribution).

\subsection{Poisson Network Model and Main Results}
\label{sec:results}

Consider a Poisson point process of uniform intensity~$\nu$ inside a network domain $\mathcal{A}$,  which determines the node positions.
The shape of the network domain does not matter, since for the upper-bound computations we will let it tend to the infinite plane to simplify the analysis.

\begin{figure}
\centering
\vskip -0.3cm
\includegraphics[width=4.9cm]{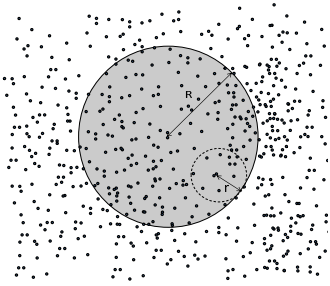}
\vskip -0.3cm
\caption{Circular cut partitioning the network domain into a disk $\D$ of radius $R$ and the remaining region $\Dc$, and a node at distance $r$ from the cut boundary.}
\label{fig:circularcut}
\end{figure}

We take an approximately circular cut of radius asymptotically equal to $R$, partitioning the network domain into two regions $\D$ and $\Dc$, as depicted in Figure~\ref{fig:circularcut} (the exact form of the cut will be clarified in Section~\ref{sec:approx_circular}).
Equivalently to Section~\ref{sec:modelbound}, we define the cut-set capacity bound $C_{\D,\Dc}$ on the communication rate achievable from nodes in $\Dc$ to nodes in $\D$. 
We denote $C_R = \E[C_{\D,\Dc}]$ the expectation of the cut-set capacity, over all node position configurations\footnote{We clarify that, for a given realization of the Poisson process, $C_R$ does not have an operational MIMO capacity meaning in the Shannon sense; it must be interpreted as the expected value of the Shannon capacity. However, we may also consider informally the case where nodes move slowly but remain Poisson distributed; then, $C_R$ is the average cut-set capacity over time.}. Our main theorem provides a simple multi-parameter bound on~$C_R$, 
as a function of $\nu$, $R$, $P$, $N$, $W$ and $\a$. The bound holds for any individual scaling behavior, 
as long as $\nu R^2 \to \infty$,
\ie
when the number of nodes in $\D$ becomes large.

\begin{theorem}
\label{th:approx_cutset}
When $\nu R^2 \to \infty$, the expected cut-set capacity is bounded by $C_R \leq \C$, with:
\[
\C \sim 2\pi \nu W \int_{\! \frac{d}{\sqrt{\nu}}}^R \log \left(1+s_r\right) (R-r) dr,
\]
where $s_r=\frac{2 \pi \nu r^{2-\a}}{\a-2}\frac{P}{N W}$, and the constant $d$ is the critical percolation radius for unit node density.
\end{theorem}

The parameter $s_r=\frac{2 \pi \nu}{\a-2} r^{2-\a} \frac{P}{N W}$,
corresponds to an upper bound on the expectation of 
the total SNR received by a given node, from all nodes at range at least $r$.

Hence, we can identify different asymptotic scaling laws, depending on the magnitude of $s_r$ 
at the upper ($r=R$) and lower ($r=\frac{d}{\sqrt{\nu}}$) limits of the integral (tending to $0$ or $\infty$). 
The derivations are detailed in~Corollary~\ref{cor:asympt}, in Section~\ref{sec:cutsetbound}.

Setting $n= \pi \nu R^2$ (the expected number of nodes in $\D$) in the latter, if we omit all the constants for simplicity, we have 
that $C_R$ is $O(\C)$, where $\C$ equals:
\[
\label{eq:asymptotics}
\left\{ 
 \begin{array}{l l l}
   W n \log(s_R),  & s_R = \omega(1) &(I)\\
    n^{2-\frac{\a}{2}}\frac{P_0}{N}, &  s_R = o(1),~\a < 3 &(II)\\
    \sqrt{n} \left(\frac{P_0}{N}\right)^{\frac{1}{\a-2}} W^{\frac{\a-3}{\a-2}}, 
&  \text{---},~\a > 3,~s_{\! \!\frac{d\,}{\sqrt{\nu}}}= \omega(1)  &(III)\\
    \sqrt{n} \frac{P_0}{N}, &    \text{---},~  \text{---},~ s_{\! \! \frac{d\,}{\sqrt{\nu}}}=o(1), &(IV)\\
  \end{array} 
\right.
\]
and $P_0=N W s_{\! \! \frac{d\,}{\sqrt{\nu}}}$
corresponds to the expected received power from nodes at range at least $\frac{d}{\sqrt{\nu}}$.

In words, we identify four different asymptotic regimes. 
When 
$s_R=\omega(1)$, the upper bound indicates that the cut-set capacity is linear in $n$ and bandwidth-limited~(I).
When $s_R=o(1)$, the capacity is power-limited and sub-linear in $n$ when $\a < 3$~(II), and both power (long-range) and bandwidth (short-range) limited when $\a>3$ and $s_{\! \! \frac{d\,}{\sqrt{\nu}}}=\omega(1)$~(III). When 
$s_{\! \! \frac{d\,}{\sqrt{\nu}}}=o(1)$, the power limitation dominates at all ranges (IV). In the two latter cases, the capacity bound is $\Theta (\sqrt{n})$.

Therefore, with $s_R$ corresponding to the long-range SNR, and $s_{\! \! \frac{d\,}{\sqrt{\nu}}}$ to the short-range SNR, 
the four described cases 
essentially map to the operating regimes identified in~\cite{snr_ot}, derived under a different perspective and methodology.
Accordingly, even though we do not compute lower bounds,
the relative asymptotic tightness of our bounds is established by comparing with these related results;
the four optimal communication schemes discussed in~\cite{snr_ot} would achieve almost order-optimal scaling performance if analyzed in our framework.

\section{Cut-set Capacity: Proof of Theorem~\ref{th:approx_cutset}}
\label{sec:cutsetbound}

From~(\ref{eq:MISOJensen}) in Section~\ref{sec:modelbound}, the cut-set capacity can be bounded above by the sum of MISO capacities, with independent transmissions at maximum power, and without fading. Hence, from now on, we assume that these conditions hold.

\subsection{MISO Bound}
\label{sec:misobound}
We consider a node $i$ at distance $r$ from the cut boundary, as depicted in Figure~\ref{fig:circularcut}. We denote $Q_r$ the total received SNR by node $i$  from all nodes in $\Dc$, and $M_r$ the MISO capacity from all nodes in $\Dc$ to $i$.
We compute upper bounds on the expectations $\E[Q_r]$ and $\E[M_r]$, over Poisson node positions.
\begin{lemma}
\label{lem:Qr}
$\E[Q_r] \leq s_r$, with
$s_r = \frac{2 \pi \nu r^{2-\a} }{\a-2} \frac{P}{N W}$.
\end{lemma}
\begin{proof}
As transmissions are independent, the expectation can be computed from Campbell's theorem~\cite[p. 28]{poisson}:
\begin{align}
\E[Q_r]  &=  \int_{\Dc}  \rho^{-\a} \frac{P}{NW} dS\nonumber\\
&\leq \int_0^{2\pi}d\phi \int_{r}^\infty \rho d\rho \cdot \nu  \rho^{-\a} \frac{P}{NW},
\end{align}
where, for the upper bound, we let $\Dc$ tend to the infinite plane, and we consider all SNR contributions from nodes at range at least~$r$ from~$i$ (instead of just the nodes in $\Dc$).
\end{proof}

\begin{lemma}
\label{lem:MISO}
$\E[M_r] \leq W \log \left(1+s_r\right)$.
\end{lemma}
\begin{proof}
From the formula for the AWGN MISO capacity, with independent transmissions and 
total received SNR~$Q_r$:
\be
M_r =  W \log(1+Q_r).
\ee
As $\log(1+x)$ is concave, we conclude using Jensen's inequality, \ie
$\E[\log(1+Q_r)] \leq \log(1+\E[Q_r]) \leq \log(1+s_r)$.
\end{proof}

\subsection{Circular Cut with Empty Outer Strip}
\label{sec:cusetsub}

We initially assume that the cut defining $\D$ and $\Dc$ is circular with radius exactly~$R$, and the outer strip of the disk~$\D$ of width $\frac{d}{\sqrt{\nu}}$ is empty. In the following lemma, we evaluate the expected cut-set capacity~$C_R$ under these assumptions.

\begin{lemma}
\label{lemth:cutset}
The expected cut-set capacity is $C_R \leq  \C$, with:
\be
\label{eq:I}
\C = 2\pi \nu W \int_\frac{d}{\sqrt{\nu}}^R \log\left(1+s_r \right) (R-r) dr.
\ee
\end{lemma}
\begin{proof}
Taking the average over Poisson node positions in~(\ref{eq:MISOJensen}), we can move the expectation inside the sum
due to the linearity of expectations (\ie $\E[X+Y] = \E[X]+\E[Y]$, even if $X$ and $Y$ are dependent random variables).

Therefore, applying Campbell's theorem (for the first moment) with the function $M(r)=\E[M_r]$:
\be
\label{eq:crint}
C_R  \leq \int_0^{2\pi} d\phi \int_0^{R-\frac{d}{\sqrt{\nu}}} \rho d\rho \cdot \nu M(R-\rho).
\ee
Finally, we
substitute Lemma~\ref{lem:MISO} into~(\ref{eq:crint}), and we change the integration variable to $r=R-\rho$.
\end{proof}

\subsection{Approximately Circular Cut}
\label{sec:approx_circular}
We now prove, using percolation theory, that there exists indeed a cut of radius approximately $R$, with an empty outer strip of width $\frac{d}{\sqrt{\nu}}$, where $d$ is the critical percolation radius for unit node density, as long as the expected number of nodes in $\D$ tends to infinity, \ie $\nu R^2 \to \infty$.

\begin{lemma}
\label{lem:percolation}
For some constant $\delta>0$, the annulus 
defined by two concentric circles of radii $R$ and $R+\frac{\delta}{\sqrt{\nu}} \log(\sqrt{\nu} R)$
contains almost surely (when $\nu R^2 \to \infty$) a vacant loop of width $\frac{k}{\sqrt{\nu}}$, for any constant $k < d$.
\end{lemma}
\begin{proof}
See appendix.
\end{proof}

To complete the proof of Theorem~\ref{th:approx_cutset}, we 
bound the expected cut-set capacity of the approximately circular cut.
\begin{proof}
The cut-set bound can be computed from Lemma~\ref{lemth:cutset}.
Since Lemma~\ref{lem:percolation} holds for any $k<d$, we can assume that the smaller distance between two nodes at opposite sides of the cut tends to $\frac{d}{\sqrt{\nu}}$.
The larger distance between opposite side nodes is $R+\frac{\delta}{\sqrt{\nu}} \log (\sqrt{\nu} R) \sim R$. Therefore, it can be verified that the integral remains asymptotically equivalent if we take $R$ as the upper limit in Lemma~\ref{lemth:cutset}.
\end{proof}

\begin{corollary}
\label{cor:asympt}
When $\nu R^2 \to \infty$, $C_R \leq \C$, with:
\[
  \C {\sim} \left\{ 
  \begin{array}{l l}
    \pi \nu R^2 W \log (\frac{\nu R^{2-\a} P}{N W}),  &s_R = \omega(1)\\
    K_1 \nu^2 R^{4-\a} \frac{P}{N}, &s_R = o(1),~\a<3\\
    4 \pi^2 \nu^2 R \log(R) \frac{P}{N}, &\text{---},~\a=3\\
    K_2 \nu^{\frac{\a-1}{\a-2}} R \left(\frac{P}{N}\right)^{\frac{1}{\a-2}} W^{\frac{\a-3}{\a-2}},&\text{---},~\a > 3,~s_{\! \! \frac{d\,}{\sqrt{\nu}}}= \omega(1)\\
    K_3 \nu^{\frac{1+\a}{2}} R \frac{P}{N}, &\text{---},~\text{---},~s_{\! \! \frac{d\,}{\sqrt{\nu}}}=o(1),\\
  \end{array} \right.
\]
and 
$K_1= \frac{4 \pi^2}{(\a-2)(3-\a) (4-\a)}$, 
$K_2= \frac{(2 \pi)^{\frac{\a-1}{\a-2}}}{(\a-2)^{\frac{1}{\a-2}}}\frac{\pi}{\sin (\frac{\pi}{\a-2})}$,
$K_3=\frac{4 \pi^2 d^{3-\a}}{(\a-2) (\a-3)}$.
\end{corollary}
\begin{proof}
See appendix.
\end{proof}

\section{Numerical Results}
\label{sec:numerical}

\begin{figure}[t!]
\centering
\hskip -1.5cm
\includegraphics[width=6.5cm]{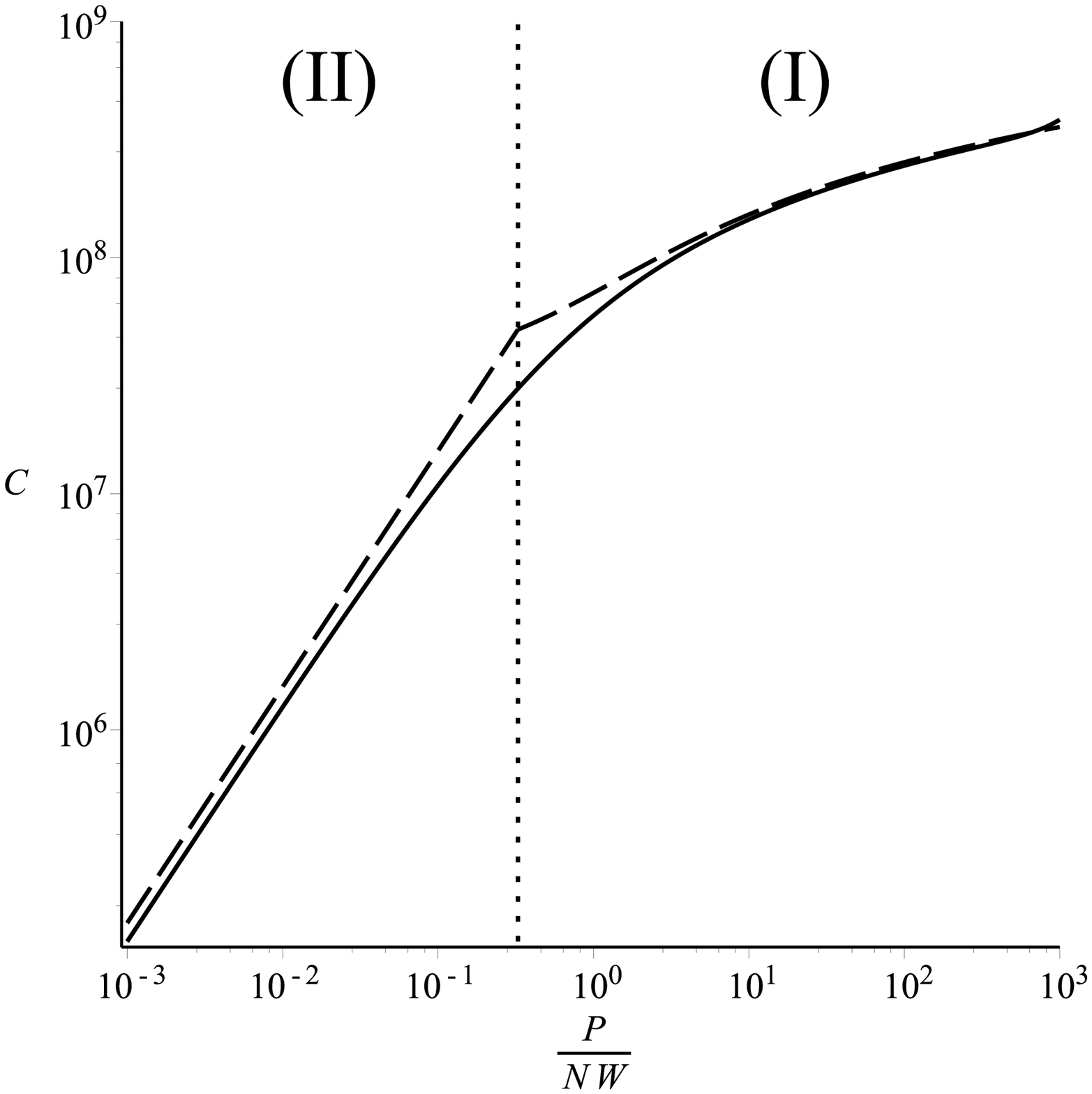}
\vskip -4.3cm
\hskip 4.6cm
\includegraphics[width=4.2cm]{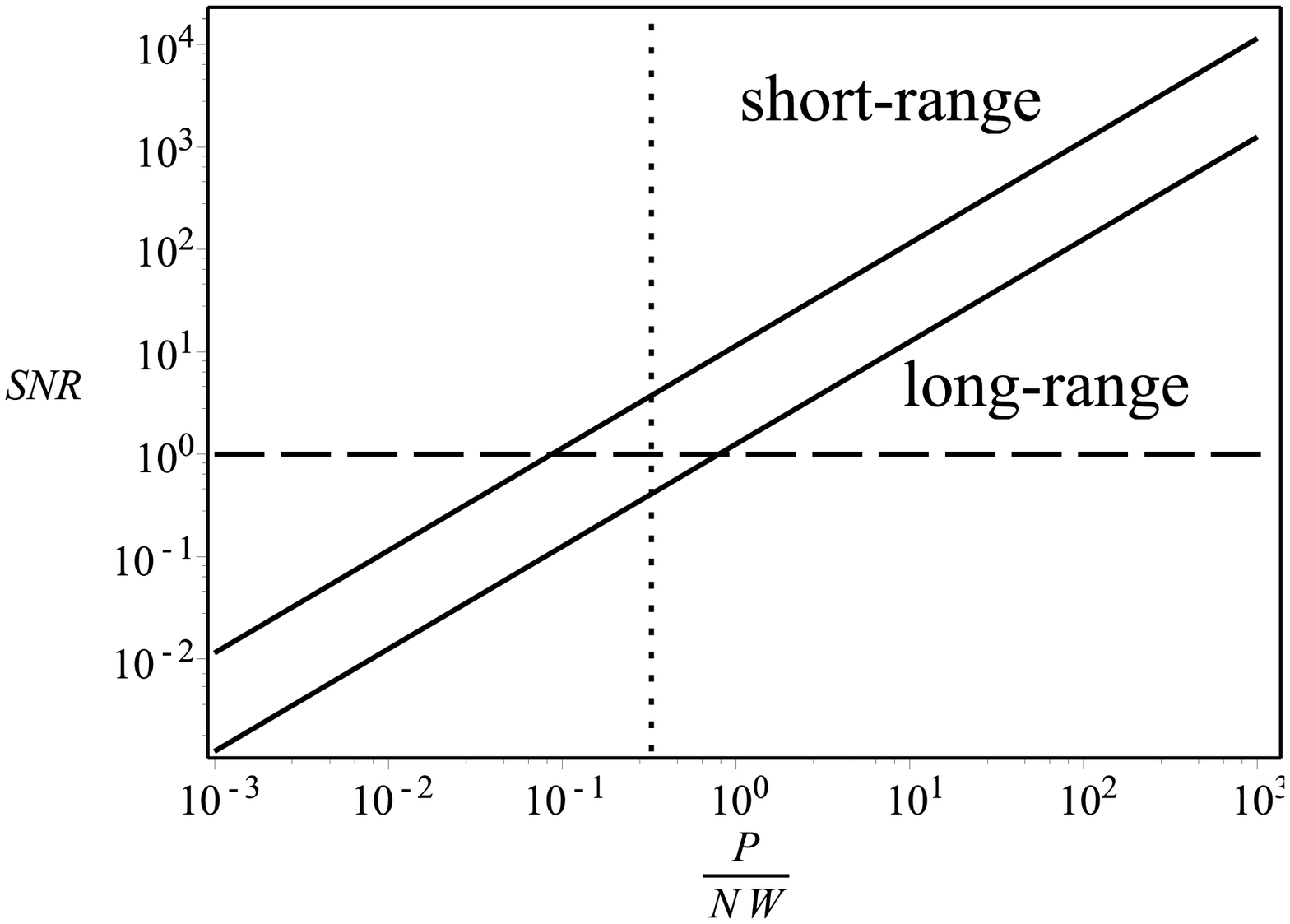}
\vskip 1.3cm
\hskip -1.5cm
\includegraphics[width=6.5cm]{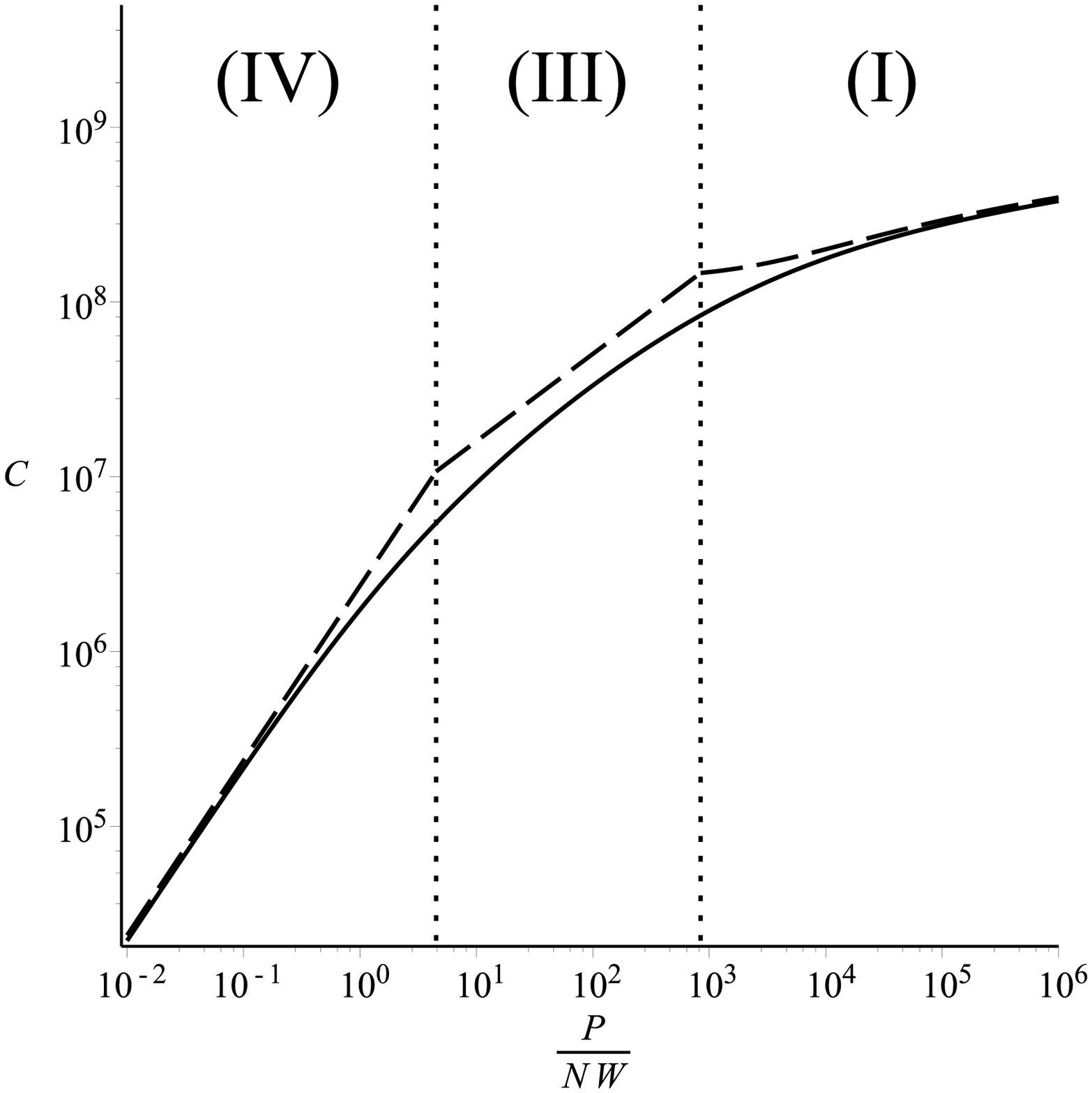}
\vskip -4.3cm
\hskip 4.6cm
\includegraphics[width=4.2cm]{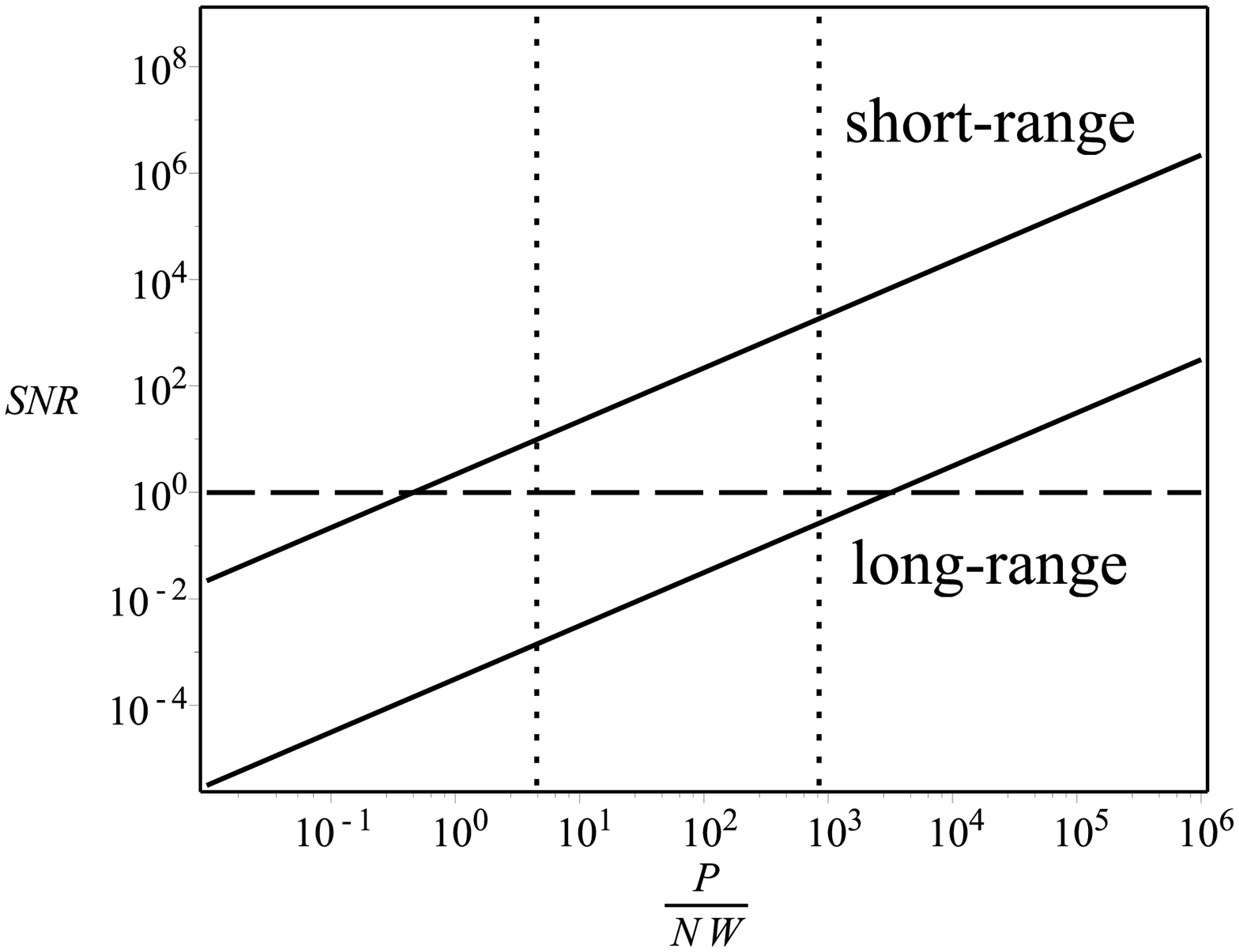}
\vskip 0.8cm
\caption{Cut-set capacity bound $C$, versus the SNR parameter $\frac{P}{N W}$, 
for $\a=2.5$ (top) and $\a=4$ (bottom), and corresponding asymptotic regimes. Small graphs plot the long-range and short-range SNRs.}
\label{fig:capacity_alpha}
\end{figure}

We provide numerical examples, illustrating the derived cut-set bounds.
Figure~\ref{fig:capacity_alpha} depicts log-log plots of the bounds on the expected cut-set capacity $C_R$ in bits per second, by varying the signal to noise parameter  $\frac{P}{N W}$, for $\a=2.5$ (top) and $\a=4$ (bottom). The remaining parameters are fixed: density $\nu=1$, radius $R=100$, bandwidth $W=10^3$. We use the numeric estimate $d=1.198$ for the percolation radius. 
The solid lines plot the upper bound from Theorem~\ref{th:approx_cutset}. The dashed lines are the asymptotic bounds in Corollary~\ref{cor:asympt} (for case~(I), we add the second-order constant factor from~(\ref{eq:Ir}) in the proof).
Insets plot the long-range ($s_R$) and short-range ($s_{\! \! \frac{d\,}{\sqrt{\nu}}}$) SNRs.

The figures illustrate the continuous transitions between the four different operating regimes we described in Section~\ref{sec:results}.
When $\a<3$, we have two different operating regimes (Figure~\ref{fig:capacity_alpha}, top).
When the long-range SNR is small, the cut-set capacity increases linearly, corresponding to the power-limited regime (II). 
When the long-range SNR becomes larger, we observe a slower logarithmic increase, and we identify the bandwidth-limited regime (I).
In contrast, when $\a>3$, the evolution of the SNR reveals three operating regimes (Figure~\ref{fig:capacity_alpha}, bottom).
When both the short-range and long-range SNRs are small, the linear capacity growth indicates the power-limited regime (IV); then, when the short-range SNR becomes large but the long-range SNR is still small, the slope changes and the regime is both power and bandwidth-limited~(III); finally, when the long-range SNR becomes large too, the capacity growth slows down considerably, as we transition to the bandwidth limited regime (I).   

\section{Conclusion}
\label{sec:conclusion}

We derived a new multi-parameter cut-set upper bound
on the capacity of wireless networks (Theorem~\ref{th:approx_cutset}), under fast fading with symmetrically distributed and independent channels (or with arbitrary channels, but under the condition that nodes may only transmit independent signals).
The asymptotic analysis (Corollary~\ref{cor:asympt})
reveals four operating regimes, which can be mapped to previously known scaling laws~\cite{snr_ot}, extending them with specific bounds.
The identification of such operating regimes is essential for the design of efficient communication strategies.

\section{Appendix}

\subsection{Proof of Corollary~\ref{cor:asympt}}

The integral in Theorem~\ref{th:approx_cutset} equals
$2 \pi \nu W \cdot  \left. I_r \right|_{\frac{d}{\sqrt{\nu}}}^R$,
where
$I_r=\int \log(1+s_r) (R-r) dr$ can be evaluated in closed form. 
We initially assume that $\{\frac{1}{\a-2}, \frac{2}{\a-2} \}\neq 1,2,\ldots$, to obtain a general formula.
This excludes $\a=3, 4$, while all other excluded values are smaller than $3$. 

We recall the definition of ordinary hypergeometric functions~\cite[Ch.  15]{nist}: ${}_2F_1 (a,b;c;x) =\sum_{n=0}^{\infty} \frac{(a)_n (b)_n}{(c)_n} \frac{x^n}{n!}$, $c \neq 0, -1, -2, \ldots$, and $(x)_n=x (x+1)\ldots(x+n-1)$ is the rising factorial (with $(x)_0=1$).
We have:
\begin{align}
\label{eq:Ir}
I_r = \log(&1+s_r) \left(R-\frac{r}{2}\right)r +(\a-2)\left(R-\frac{r}{4}\right) r\nonumber\\
&- (\a-2) \left(R g_1(s_r)-\frac{r}{4} g_2(s_r) \right) r,
\end{align}
with $g_k(x) = {}_2F_1\left(1,-\frac{k}{\a-2};1-\frac{k}{\a-2}; -x\right), k=\{1,2\}$.

The definitions of $g_k(s_r)$ as hypergeometric functions yield full asymptotic expansions for $s_r=o(1)$.

Using a linear transformation~\cite[eq. 15.8.2]{nist}, we obtain full asymptotic expansions for $s_r=\omega(1)$:
\begin{align}
g_k(s_r)=&\frac{1}{s_r} \frac{k\cdot {}_2F_1\left(1,1+\frac{k}{\a-2};2+\frac{k}{\a-2}; -\frac{1}{s_r}\right) }{k+\a-2}\nonumber\\
&+ s_r^{\frac{k}{\a-2}} \frac{k \pi}{(\a-2) \sin (\frac{k \pi}{\a-2})},~~k=\{1,2\}.
\label{eq:sd}
\end{align}

From~(\ref{eq:Ir}) and~(\ref{eq:sd}), the asymptotic analysis of $I_r$  is straightforward. 
The main terms for both integration limits are:
\begin{align}
I_R &\sim \left\{ 
\begin{array}{l l}
\frac{R^2}{2} \log(s_R), \quad &s_R=\omega(1)\\ 
\frac{2 \pi \nu R^{4-\a}}{(\a-2)(3-\a) (4-\a)} \frac{P}{N W}, ~&s_R=o(1)
\end{array}
\right.
\label{eq:IR}
\\
I_{\! \! \frac{d\,}{\sqrt{\nu}}} &\sim \left\{ 
\begin{array}{l l}
 R \left(\frac{2 \pi \nu}{\a-2} \right)^{\frac{1}{\a-2}} \frac{\pi}{\sin (\frac{\pi}{\a-2})}, \quad \quad &s_{\! \! \frac{d\,}{\sqrt{\nu}}}=\omega(1)\\ 
\frac{2 \pi \nu^{\frac{\a-1}{2}} R d^{3-\a}}{(\a-2)(3-\a)} \frac{P}{N W}, \quad \quad &s_{\! \! \frac{d\,}{\sqrt{\nu}}}=o(1).
\end{array}
\right.
\label{eq:Id}
\end{align}

When $s_R = \omega(1)$, the main asymptotic term is always $I_R$, \ie the first case of~(\ref{eq:IR}).

When $s_R=o(1)$ and $\a < 3$, the main asymptotic term is again $I_R$, now equal to the second case in~(\ref{eq:IR}).

When $s_R=o(1)$ and $\a > 3$, the main asymptotic term is~$I_{\! \! \frac{d\,}{\sqrt{\nu}}}$. So, we obtain the two cases of~(\ref{eq:Id}), when $s_{\! \! \frac{d\,}{\sqrt{\nu}}}= \omega(1)$ and $s_{\! \! \frac{d\,}{\sqrt{\nu}}}=o(1)$, respectively.

For completeness, we consider the excluded values of~$\a$. 
For $\a=3$,
$\a=4$, a simple integration confirms 
Corollary~\ref{cor:asympt}.
For the remaining cases, we have $\a<3$, 
and the 
main asymptotic terms are the same as in the general case.
It suffices to note that, when $s_R=\omega(1)$, we can use $\log(1+s_r)=\log(s_r)+o(1)$ to recover the main asymptotic term: $\pi \nu R^2 \log(s_R)+O(\nu R^2)$.
When $s_R=o(1)$, we use the fact that $\log(1+s_r) \leq s_r$ to perform the integration on $s_r$. The result is asymptotically tight; for $\a<3$, the upper limit $I_R$ is always dominant, and indeed $\log(1+s_r) = s_r +o(s_r)$ when $r=\Theta(R)$.

\subsection{Proof of Lemma~\ref{lem:percolation}}

We consider the Boolean continuum percolation model \cite{roy} where nodes are placed with Poisson intensity $\nu$, and they are connected within distance $x$.
The critical percolation radius
is~$\frac{d}{\sqrt{\nu}}$, where
$d$ is the critical radius with unit node density.

We follow the definition of vacant and occupied regions from~\cite[p. 15]{roy}. 
We consider an annulus of inner perimeter $\ell$ and width $m$. Let $\P(\text{vacant-loop})$ be the probability that the annulus contains a vacant loop of width $x$.
Let $\P(\text{TB-occupied})$ be the probability that the annulus contains an occupied top-bottom crossing connecting the two circular sides. 
Clearly,
\be
\label{eq:vac}
\P(\text{vacant-loop}) =1 - \P(\text{TB-occupied}).
\ee

Let $p_i$, with $i=1,2,\ldots$, and $0 \leq i < \frac{\ell}{x}$ be a sequence of points on the inner annulus boundary, at equal distance~$x$ (except possibly the two points closing the circle, which are at distance at most $x$).
For the existence of an occupied component in the direction of width $m$, there must be at least one occupied component of diameter at least $m$, from some~$p_i$.
The probability that a connected component of diameter at least $m$ exists, is bounded by the location invariant probability that there is a connected path from the origin to the boundary of a square box $[-m,m] \times [-m,m]$, centered at the origin, which we denote: $\P(0 \overset{o}{\leadsto} \partial B_m)$.
Taking a union bound:
\be
\label{eq:occ}
\P(\text{TB-occupied}) \leq \frac{\ell}{x} \cdot \P(0 \overset{o}{\leadsto} \partial B_m).
\ee

With appropriate scaling of the distances by $\sqrt{\nu}$ (to account for a node density $\nu$ instead of $1$), Theorem~2.4 in~\cite{roy} implies that,
for any $x = \frac{k}{\sqrt{\nu}}$, where $k$ is a constant independent of $\nu$ such that $k<d$,
\be
\label{eq:pom}
\P(0 \overset{o}{\leadsto} \partial B_m) \leq  c_1 e^{- c_2 m \sqrt{\nu}},
\ee
for some constants $c_1,c_2 >0$ depending on $k$.

Therefore, from~(\ref{eq:vac}), (\ref{eq:occ}) and~(\ref{eq:pom}), we obtain the bound:
\be
\P(\text{vacant-loop}) \geq 1 - c_1 \frac{\ell \sqrt{\nu}}{k} e^{-c_2 m \sqrt{\nu}}.
\ee
Taking $\ell=2 \pi R$, and $m= \frac{\delta}{\sqrt{\nu}} \log(R \sqrt{\nu})$, with $\delta>\frac{1}{c_2}$,
\be
\P(\text{vacant-loop}) \geq1 - \frac{c_1 2 \pi}{k} \left(R \sqrt{\nu}\right)^{1- \delta c_2}
\underset{\nu R^2 \to \infty}{\to} 1.
\ee

\end{document}